\documentclass[]{llncs}

\usepackage[utf8]{inputenc}

\usepackage{graphicx}

\usepackage{amsmath}

\newcommand{\mC}{\mathcal{C}}

\spnewtheorem{algo}{Algorithm}{\bfseries}{\normalfont}
\spnewtheorem{obs}{Observation}{\bfseries}{\itshape}

\title{3-Colourability of Dually Chordal Graphs in Linear Time}

\author{Arne Leitert}

\institute{
    University of Rostock, Germany \\
    \email{arne.leitert@uni-rostock.de}
}

\begin{document}

\pagestyle{plain}

\maketitle

\begin{abstract}
    A graph $G$ is dually chordal if there is a spanning tree $T$ of $G$ such that any maximal clique of $G$ induces a subtree in $T$. This paper investigates the Colourability problem on dually chordal graphs. It will show that it is NP-complete in case of four colours and solvable in linear time with a simple algorithm in case of three colours. In addition, it will be shown that a dually chordal graph is 3-colourable if and only if it is perfect and has no clique of size four.
\end{abstract}

\section{Introduction}
Colouring a graph is the problem of finding the minimal number of colours required to assign each vertex a colour such that no adjacent vertices have the same colour. It is a classic problem in computer science and one of Karp's~21 NP-complete problems \cite{Karp72}. Colouring remains NP-complete if bounded to three colours \cite{Stockmeyer1973}. Also, if a 3-colourable graph is given, it remains NP-hard to find a colouring, even if four colours are allowed \cite{4col3colNPhard}.
  
Dually chordal graphs are closely related to chordal graphs, hypertrees and $\alpha$-acyclic hypergraphs, which also explains the name of this graph class (see \cite{duallyChordal} for more). Section~\ref{sec:DuallyChordal} will show, that Colouring is NP-complete for dually chordal graphs if bounded to four colours. So the aim of this paper is to find a linear time algorithm for the 3-Colouring problem.

Next to the classic vertex colouring problem, there are the problems of colouring only edges or edges and vertices of a graph. For dually chordal graphs they have been investigated by de~Figueiredo et~al.~\cite{deFigueiredo1999147}.

\section{Basic Notions}\label{sec:BasicNotions}
Let $G=(V,E)$ be a graph with the vertex set $V$, the edge set $E$. For this paper, any graph is finite, undirected, connected  and without loops or multiple edges. For a set $U \subseteq V$, $G[U]$ denotes the induced subgraph of $G$ with the vertex set $U$. Let $\overline{G}=(V,\overline{E})$ be the complement of $G$, such that $\overline{E} = \{ uv \mid u,v \in V; u \neq v; uv \notin E \}$.

A set of vertices $S \subseteq V$ is a \emph{clique} if for each pair $u,v \in S$ ($u \neq v$), $u$ and $v$ are adjacent. A clique of size $i$ is denoted as $K_i$. The number of vertices in a maximum clique in $G$ is the \emph{clique number $\omega(G)$} of $G$. Let a \emph{diamond} be a graph with four vertices and five edges. The edge connecting the vertices with degree 3 is called \emph{mid-edge}. A \emph{chordless cycle} $C_k$ has $k$ vertices $v_1, \ldots, v_k$ and the edges $v_iv_{i+1}$ (index arithmetic modulo $k$). If a chordless cycle has at least five vertices, it is also called $hole$. A \emph{wheel} $W_k$ is a $C_k$ plus a vertex $v_w$ and the edges $v_wv_i$ ($1 \leq i \leq k$). A chordless cycle $C_k$ or wheel $W_k$ is \emph{odd} if $k$ is odd.

Let $N(v) := \{ u \in V \mid uv \in E \}$ denote the \emph{open} and $N[v] := N(v) \cup \{v\}$ the \emph{closed neighbourhood} of the vertex $v$. A graph is \emph{locally connected} if for all vertices $v$ the open neighbourhood $N(v)$ is connected.

A vertex $u \in N[v]$ is a \emph{maximum neighbour} of $v$, if for all $w \in N[v]$, $N[w] \subseteq N[u]$ holds. Note that $u = v$ is not excluded. A vertex ordering $(v_1, \ldots, v_n)$ is a \emph{maximum neighbourhood ordering} if every $v_i$ ($1 \leq i \leq n$) has a maximum neighbour in $G[\{v_i, \ldots, v_n \}]$.

A vertex $v$ is an articulation point of $G$ if $G[V \setminus \{v\}]$ is not connected. If a graph has no articulation point it is \emph{biconnected} (or \emph{2-connected}). Maximal biconnected subgraphs are called \emph{blocks}. Any connected graph can decomposed into a block tree. If two blocks intersect, they have exactly one common vertex which is an articulation point.
 
\bigskip

An \emph{independent set} is a vertex set $I \subseteq V$ such that for all vertices $u,v \in I$, $uv \notin E$ holds. A graph is \emph{$k$-colourable} if $V$ can be partitioned into $k$ independent sets $V_1, \ldots, V_k$ with $V = V_1 \cup \ldots \cup V_k$ and $V_i \cap V_j = \emptyset$ ($i \neq j$). The \emph{chromatic number} $\chi(G)$ of a graph $G$ is the lowest $k$ such that $G$ is $k$-colourable. The \emph{$k$-Colourability} problem asks if a graph is $k$-colourable. Accordingly the \emph{Colourability} problem asks for the chromatic number of a graph. It is easy to see, that 2-Colourability can be solved in linear time for each graph.

A graph $G$ is \emph{perfect} if for each induced subgraph $G^*$ of $G$ the chromatic and the clique number are equal ($\omega(G^*) = \chi(G^*)$). Chudnovsky et~al.~\cite{StrongPerfectGraph} have shown that a graph is perfect if and only if it is $\left( C_{2n+5}, \overline{C_{2n+5}} \right)$-free for all $n \geq 0$. This is known as the Strong Perfect Graph Theorem.  Perfect graphs can be recognised and coloured in polynomial time \cite{RecogBerge}\cite{ColorPerfect}.

For a graph $G$, $K(G)$ denotes the \emph{clique graph} of $G$, where each vertex in $K(G)$ represents a maximal clique of $G$ and two vertices are connected if the corresponding cliques have a common vertex, i.e. $K(G)=(V, E)$ with $V = \{ k \mid \text{$k$ is a maximal clique of $G$}\}$ and $E= \{ k_1k_2 \mid k_1 \cap k_2 \neq \emptyset \}$. A graph $G$ is \textit{chordal} if it is $C_k$-free ($k \geq 4$). A graph is \emph{clique-chordal} if its clique graph is chordal.

\section{Dually Chordal Graphs}\label{sec:DuallyChordal}
Dually chordal graphs are originally defined by a maximum neighbourhood ordering.

\begin{definition}
    A graph is \emph{dually chordal} if and only if it has a maximum neighbourhood ordering.
\end{definition}

Brandstädt et~al.~\cite{duallyChordal} give an overview of characterisations for dually chordal graphs. One of it allows to recognise graphs of this class in linear time (with an algorithm by Tarjan and Yannakakis~\cite{RecogAlgor}). Another characterisation was developed in~\cite{LeitertDiploma2012}:
 
\begin{lemma}[\cite{duallyChordal}, \cite{LeitertDiploma2012}]\label{theo:dcDef}
    Let $P_T(u,v)$ the set of vertices on the path from $u$ to $v$ in the tree $T$, with $u,v \notin P_T(u,v)$. For a graph $G=(V,E)$ the following conditions are equivalent:
    \begin{enumerate}
        \item  $G$ is dually chordal.
        \item \label{item:dcDefCliqueTree} There is a spanning tree $T$ of $G$ such that every maximal clique of $G$ induces a subtree in $T$.
        \item \label{item:dcDefPathTree} There is a spanning tree $T$ of $G$ such that for every edge $uv \in E$ the following holds: $\forall\,w \in P_T(u,v): uw, vw \in E$
    \end{enumerate}
\end{lemma}

\begin{lemma}[\cite{LeitertDiploma2012}]\label{lem:CliqueTreeEqualPathTree}
	 The spanning trees in conditions~\ref{item:dcDefCliqueTree} and \ref{item:dcDefPathTree} of Lemma~\ref{theo:dcDef} are equal. 
\end{lemma}

Since $P_T(u,v)$ does not include the vertices $u$ and $v$, (according to the notation for the open and closed neighbourhood) let $P_T[u,v] := P_T(u,v) \cup \{u,v\}$.

\begin{corollary}
    If $uv \in E$ then $P_T[u,v]$ is a clique.
\end{corollary}

Is is easy to see that an arbitrary graph becomes dually chordal if a vertex adjacent to all vertices is added. This leads to a reduction for several NP-complete problems, including Colourability. The idea for this method was already given by Brandstädt et~al.~\cite{Brandstaedt199843}.

\begin{theorem}\label{theo:4colDC}
	4-Colourability is NP-complete for dually chordal graphs.
\end{theorem}

\begin{proof}
	Let $G=(V,E)$ be an arbitrary graph. Thus, $G_{dc}=(V \cup \{ u \}, E \cup \{uv \mid v \in V \})$ with $u \notin V$ is dually chordal. It follows that $G$ is 3-colourable if and only if $G_{dc}$ is 4-colourable. Because 3-Colourability is NP-complete in general~\cite{Stockmeyer1973}, 4-Colourability is NP-complete for dually chordal graphs.
	\qed
\end{proof}

\section{$K_4$-free Dually Chordal Graphs}\label{sec:K4Free}
Obviously a graph has to be $K_4$-free to be 3-colourable. So this section investigates $K_4$-free dually chordal graphs.

For this section let $G=(V,E)$ be a $K_4$-free dually chordal graph. The spanning tree $T=(V,E_t)$ of $G$ and the set $P_T[u,v]$ are defined as in condition~\ref{item:dcDefPathTree} of Theorem~\ref{theo:dcDef} and the paragraph after Lemma~\ref{lem:CliqueTreeEqualPathTree}.

\begin{corollary}\label{cor:K4FreePEdges}
    If $uv \in E$ then $|P_T[u,v]| \leq 3$, i.e. $uv \in E \land uv \notin E_t \Rightarrow \exists!\, w: uw,wv \in E_t$.
\end{corollary}

For the next lemma the notion of a \emph{route} is defined:

\begin{definition}
    A \emph{route} on $G$ is a list of vertices $(v_1,\ldots,v_k)$ such that $v_iv_{i+1} \in E$ or $v_i=v_{i+1}$ for all $i < k$.
\end{definition}

\begin{lemma}\label{lem:CircleWheel}
    Let $\mC$ be a $C_k$ in $G$ with $k \geq 4$. There is a vertex $w$ such that $\mC$  and $w$ form a wheel. Also no edge of $\mC$ is in $T$.
\end{lemma}

\begin{proof}
    Let $\mC$ be a chordless cycle with the vertices $\{c_1, \ldots, c_k\}$ ($k \geq 4$) and the edges $c_ic_{i+1}$ (index arithmetic modulo $k$). Based on Corollary~\ref{cor:K4FreePEdges} for every edge there is the set of vertices $P_T[c_i, c_{i+1}] = \{ c_i, w_i, c_{i+1} \}$ with $c_i = w_i$ if and only if $c_ic_{i+1} \in E_t$. Thus, it is possible to build a route $\rho=(c_1,w_1,c_2,\ldots,w_k,c_1)$ on $T$ (see Figure~\ref{fig:circleWheel}).

\begin{figure}[htb]
    \centering
    \includegraphics{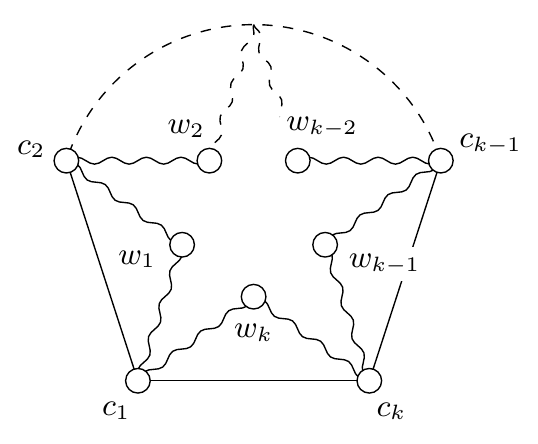}
    \caption{The circle $\mC$ and the route $\rho$ (waved).}
    \label{fig:circleWheel}
\end{figure}

Because all the edges of $\rho$ are in $T$, $\rho$ also induces a tree. Therefore $\rho$ can be eliminated by removing leaves. Let $(\ldots, v_{i-1}, v_i, v_{i+1}, \ldots)$ be a subroute of $\rho$ and $v_i$ a leaf. Thus, if $v_{i-1} \neq v_i$ and $v_i \neq v_{i+1}$ then $v_{i-1} = v_{i+1}$, so $v_{i-1}$ or $v_{i+1}$ can also removed from $\rho$ with $v_i$.

Now assume that no vertex of $\mC$ is a leaf. In this case a vertex $w_i$ is a leaf. So it follows, that $c_i = c_{i+1}$. Therefore there must be a leaf $c_i$ and also $w_{i-1}=w_i$ (otherwise $c_{i-1}c_{i+1} \in E$). This allows to remove $c_i$ and $w_i$ from $\rho$.

Continuing this procedure, it follows that $w_1=\ldots=w_k$ and $c_i \neq w_i$ for all $i$. Thus, no edge $c_ic_{i+1}$ is an edge of $T$ and there is a vertex $w$, such that $\mC$ and $w$ form a wheel.
\qed
\end{proof}

\begin{theorem}\label{theo:3colIffPerfect}
    A dually chordal graph $G$ is 3-colourable if and only if $G$ is perfect and $K_4$-free.
\end{theorem}

\begin{proof}
$\Leftarrow:$ By definition of perfect graphs.

$\Rightarrow:$ $G$ is 3-colourable, so it is $K_4$-free. Let $k \geq 7$. It is easy to see, that no $\overline{C_k}$ is 3-colourable. Therefore, $G$ has no $\overline{C_k}$. From the 3-colourability it also follows that there is no odd wheel in $G$ and by Lemma~\ref{lem:CircleWheel} no odd hole\footnote{also no $\overline{C_6}$}. Thus, by the Strong Perfect Graph Theorem, $G$ is perfect (recall, $C_5 = \overline{C_5}$).
\qed
\end{proof}

Theorem~\ref{theo:3colIffPerfect} now leads to a polynomial time algorithm for 3-Colourability of dually chordal graphs.

\section{Linear Time Algorithm}\label{sec:LinTimeAlgo}
After showing that there is a polynomial time algorithm, this section presents a linear time algorithm for 3-Colourability of dually chordal graphs based on the following theorem:

\begin{theorem}\label{theo:locConConstruction}
    Each connected locally connected graph can be constructed by starting with an arbitrary edge and then add only vertices that has at least two adjacent neighbours in the already constructed graph.
\end{theorem}

\begin{proof}
Let $(u, v, v_1, \ldots, v_k)$ be the ordering the vertices are added to construct a locally connected graph $G=(V,E)$ starting with an edge $uv$. Also let $V_i := \{u, v, v_1, \ldots, v_i \}$ with $V_0 := \{ u, v \}$ and $0 \leq i \leq k$.

Now let $G_i := G[V_i]$ ($0 \leq i < k$) be a locally connected subgraph of $G$. Assume, there is no vertex~$w$ adjacent to two vertices in $G_i$. Then $G$ is not connected or $w$ has only one neighbour~$v_n$ in $G_i$. In second case, because $G$ is locally connected, there is a vertex set $S \subseteq N(v_n) \setminus V_i$ connecting $w$ with the neighbours of $v_n$ that are already in $G_i$. Thus, there is a vertex $s \in S$ adjacent to $v_n$ and a neighbour of $v_n$ in $G_i$. Therefore, set $v_{i+1} := s$.
\qed
\end{proof}

Theorem~\ref{theo:locConConstruction} allows to colour graphs whose blocks are locally connected by using the following strategy: Select an uncoloured vertex with a minimal number of available colours and give it an available one. Repeat this until all vertices are coloured or there is no available colour for a vertex. Algorithm~\ref{algo:3colLocCon} describes this more detailed:

\begin{algo}\label{algo:3colLocCon}
\ \\
\textbf{Input:} A graph $G=(V,E)$ whose blocks are locally connected.\\
\textbf{Output:} A 3-colouring for $G$ if and only if $G$ is 3-colourable.

\begin{enumerate}
    \item Give every vertex $v$ a set of available colours $c(v) := \{1,2,3\}$. Also every vertex gets the possibility to get marked as coloured. At the beginning no vertex is marked. 

    \item \textbf{While} There are uncoloured vertices
    
    \begin{enumerate}
        \item Select an uncoloured vertex~$v$ for which $|c(v)|$ is minimal.
        \item If $c(v) = \emptyset$, \textbf{STOP:} $G$ is not 3-colourable.
        \item Disable all but one colour in $c(v)$ and mark $v$ as coloured.
        \item Disable $c(v)$ for all neighbours of $v$ ($\forall\,u \in N(v): c(u) := c(u) \setminus c(v)$).
    \end{enumerate}
\end{enumerate}
\end{algo}

\begin{theorem}\label{theo:3colLocConAlog}
    Algorithm~\ref{algo:3colLocCon} works correctly and can be implemented in linear time.
\end{theorem}

\begin{proof}
\ \par \emph{Correctness}.
After selecting the first two vertices, the algorithm only selects vertices with two adjacent neighbours, i.e. with $|c(v)| \leq 1$, until all vertices of a locally connected block are coloured (Theorem~\ref{theo:locConConstruction}). Therefore, the colouring is unique for a locally connected block.

Also, if a vertex~$v$ with $|c(v)| = 2$ and the coloured neighbour~$u$ is selected, $u$ is an articulation point. This allows to treat $u$ and $v$ as the first coloured vertices of a locally connected graph. Therefore, each available colour for $v$ leads to a correct colouring.

\medskip

\emph{Complexity}.
It is easy to see, that line~1 runs in linear time and line~2(b) in constant time such as line~2(c). Line~2(d) is bounded by the number of neighbours. Because $\sum_{v \in V} |N(v)| = 2|E|$, the full time needed (overall iterations) for line~2(d) is in ${\cal O}(|E|)$.

Line~2(a) can be implemented by using three doubly linked lists as queues. The lists include the uncoloured vertices. One is for vertices with $|c(v)| = 3$, one for $|c(v)| = 2$ and one for $|c(v)| \leq 1$. If there is also a pointer from each vertex to its entry in the list, adding and removing a vertex can be done in constant time. Thus, line~2(a) runs in constant time.
\qed
\end{proof}

Another linear time algorithm for 3-Colourability of locally connected graphs was presented by Kochol~\cite{LinearThreeColoring}, but is more complicated.

\bigskip

Dually chordal graphs are a subclass of clique-chordal graphs \cite{duallyChordal}. The next lemma will show, that a 3-colouring for clique-chordal graphs can be be found with Algorithm~\ref{algo:3colLocCon}.

\begin{lemma}
    Each block of a clique-chordal graphs is locally connected.
\end{lemma}

\begin{proof}
	Let $G=(V, E)$ be a clique-chordal graph and $K=(V_k, E_k)$ its clique graph. Assume that $G$ has a block which is not locally connected. Then there is a vertex~$v$ such that its neighbourhood is not connected and it is no articulation point. Therefore, there is a chordless cycle ${\cal C}$ fulfilling the conditions that there is no path in $G[N(v)]$ connecting the two neighbours of $v$ in ${\cal C}$. 
	
	Let ${\cal C}_G$ be the smallest ${\cal C}$ and using the following notation: $u$, $v$, $w$ and $a$, $b$, $c$ are vertices of ${\cal C}_G$ with $\{ uv, vw, ab, cd \} \subseteq E$, $u \notin \{ b, c \}$, $v \notin \{a, b, c \}$, $w \notin \{ a, b \}$, $a \notin \{ v, w \}$, $b \notin \{ u, v, w \}$ and $c \notin \{ u, v \}$. If ${\cal C}_G$ is a $C_4$, then $u=a$ and $w=c$. Figure~\ref{fig:cycleProofCliqueChordal} illustrates this notation.
	
\begin{figure}[htb]
    \centering
%
%
%
%
%
%
%
    \includegraphics{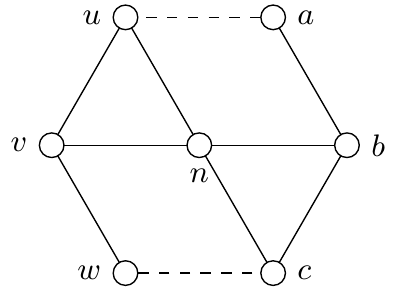}
    \caption{The cycle ${\cal C}_G$ and the vertex~$n$.}
    \label{fig:cycleProofCliqueChordal}
\end{figure}

    There is no clique in $G$ that contains two edges of ${\cal C}_G$, otherwise ${\cal C}_G$ would not be chordless. Thus, each edge represents another clique of $G$. Let $k_{uv}$, $k_{vw}$, $k_{ab}$ and $k_{bc}$ denote these cliques. Because the cliques have a common vertex if the corresponding edges have one, they are forming a cycle ${\cal C}_K$ in $K$.
    
    $K$ is chordal. Therefore, ${\cal C}_K$ has at least one chord. Without loss of generality, let $k_{uv}k_{bc}$ be such a chord. Thus, $k_{uv}$ and $k_{bc}$ have a common vertex~$n$ that is adjacent to $u$, $v$, $b$ and $c$.

    If $w = c$, $n$ is adjacent to $w$. If $w \neq c$, the sequence $(v, n, c, \ldots, w)$ induces a chordless cycle smaller than ${\cal C}_G$. In both cases ${\cal C}_G$ is not the smallest cycle~${\cal C}$ such that there is no path in $G[N(v)]$ connecting the two neighbours of $v$ in~${\cal C}$. 
    
    Therefore, $G$ is nearly locally connected. \qed
\end{proof}

It follows:

\begin{corollary}
	The 3-Colourability problem can be solved in linear time for dually chordal and clique chordal graphs.
\end{corollary}

\section{Conclusion}
After introducing dually chordal graphs and investigating the connection between $K_4$-free dually chordal and perfect graphs, this paper presented a linear time algorithm to find a 3-colouring by using the structure of locally connected graphs. This algorithm also computes
a correct 3-colouring for each graph whose blocks are locally connected.

\subsubsection{Acknowledgement.}
The author is grateful to H.\,N. de~Ridder for stimulating discussions.


\begin{thebibliography}{10}
\providecommand{\url}[1]{\texttt{#1}}
\providecommand{\urlprefix}{URL }

\bibitem{Brandstaedt199843}
Brandstädt, A., Chepoi, V.D., Dragan, F.F.: The algorithmic use of hypertree
  structure and maximum neighbourhood orderings. Discrete Applied Mathematics
  82(1–3),  43--77 (1998)

\bibitem{duallyChordal}
Brandstädt, A., Dragan, F., Chepoi, V., Voloshin, V.: Dually chordal graphs.
  SIAM Journal on Discrete Mathematics  11(3),  437--455 (1998)

\bibitem{RecogBerge}
Chudnovsky, M., Cornuéjols, G., Liu, X., Seymour, P., Vušković, K.:
  Recognizing Berge Graphs. Combinatorica  25,  143--186 (March 2005)

\bibitem{StrongPerfectGraph}
Chudnovsky, M., Robertson, N., Seymour, P., Thomas, R.: The strong perfect
  graph theorem. Annals of Mathematics  164,  51--229 (2006)

\bibitem{deFigueiredo1999147}
de~Figueiredo, C.M., Meidanis, J., de~Mello, C.P.: Total-chromatic number and
  chromatic index of dually chordal graphs. Information Processing Letters
  70(3),  147--152 (1999)

\bibitem{ColorPerfect}
Grötschel, M., Lovász, L., Schrijver, A.: Polynomial algorithms for perfect
  graphs. In: Berge, C., Chvátal, V. (eds.) Topics on Perfect Graphs,
  North-Holland Mathematics Studies, vol.~88, pp. 325--356. North-Holland
  (1984)

\bibitem{4col3colNPhard}
Guruswami, V., Khanna, S.: On the hardness of 4-coloring a 3-colorable graph.
  SIAM Journal on Discrete Mathematics  18(1),  30--40 (2004)

\bibitem{Karp72}
Karp, R.: Reducibility among combinatorial problems. In: Miller, R., Thatcher,
  J. (eds.) Complexity of Computer Computations, pp. 85--103. Plenum Press
  (1972)

\bibitem{LinearThreeColoring}
Kochol, M.: Linear algorithm for 3-coloring of locally connected graphs. In:
  Jansen, K., Margraf, M., Mastrolilli, M., Rolim, J. (eds.) Experimental and
  Efficient Algorithms, pp. 191--194. No. 2647 in Lecture Notes in Computer
  Science, Springer Berlin / Heidelberg (2003)

\bibitem{LeitertDiploma2012}
Leitert, A.: {Das Dominating Induced Matching Problem für azy\-k\-lische
  Hypergraphen}. Diploma thesis, University of Rostock, Germany (2012), in
  German

\bibitem{Stockmeyer1973}
Stockmeyer, L.: Planar 3-colorability is polynomial complete. SIGACT News
  5(3),  19--25 (Jul 1973)

\bibitem{RecogAlgor}
Tarjan, R.E., Yannakakis, M.: {Simple linear-time algorithms to test chordality
  of graphs, test acyclicity of hypergraphs, and selectively reduce acyclic
  hypergraphs.} SIAM J. Comput.  13,  566--579 (1984)

\end{thebibliography}
\end{document}